\documentclass[10pt,conference,twocolumn,twoside]{IEEEtran}
\ifCLASSINFOpdf
  \else
\fi

\usepackage{xcolor}
\usepackage{amsmath}
\usepackage{epsfig}
\usepackage{graphics}
\usepackage{psfrag}
\usepackage{epstopdf}
\usepackage{amssymb}
\usepackage{cases}
\usepackage{bm}
\usepackage{mathrsfs}
\usepackage[center]{caption}
\usepackage{subfig}
\usepackage{flushend}
\usepackage{color}
\usepackage{url}
\usepackage[francais, english]{babel}
\newtheorem{definition}{Definition}
\newtheorem{proposition}{Proposition}

\begin{document}
\title{Energy-Efficient Spectrum Sharing in Relay-Assisted Cognitive Radio Systems}
\author{\authorblockN{Mariem Mhiri\authorrefmark{1}, Karim Cheikhrouhou\authorrefmark{1}, Abdelaziz Samet\authorrefmark{1}, Fran\c{c}ois M\'{e}riaux\authorrefmark{2} and Samson Lasaulce\authorrefmark{2}} 
\vspace{0.2cm}
\authorblockA{\authorrefmark{1}Tunisia Polytechnic School P.B. 743-2078, University of Carthage, La Marsa, Tunisia\\
\small{\{mariem.mhiri, karim.cheikhrouhou\}@gmail.com, abdelaziz.samet@ept.rnu.tn}}
\vspace{0.2cm}
\authorblockA{\authorrefmark{2}L2S - CNRS - SUPELEC, F-91192 Gif-sur-Yvette, University of Paris-Sud, France\\
\small{\{meriaux, lasaulce\}@lss.supelec.fr}
}}


\IEEEpeerreviewmaketitle
\IEEEspecialpapernotice{(Invited Paper)}
\maketitle
\selectlanguage{english}
\begin{abstract}
This work characterizes an important solution concept of a relevant spectrum game. Two energy-efficient sources communicating with their respective destination compete for an extra channel brought by a relay charging the used bandwidth through a pricing mechanism. This game is shown to possess a unique Nash bargaining solution, exploiting a time-sharing argument. This Pareto-efficient solution can be implemented by using a distributed optimization algorithm for which each transmitter uses a simple gradient-type algorithm and alternately updates its spectrum sharing policy. Typical numerical results show to what extent spectral efficiency can be improved in a system involving selfish energy-efficient sources.

\end{abstract}
\section{Introduction}
Designing spectrally efficient communication systems has always been, and still is, a critical issue in wireless networks. The need for energy-efficient terminals at both the mobile and fixed infrastructure sides is more recent but becomes stronger and stronger. This paper precisely considers both aspects. More specifically, the main goal is to determine an energy-efficient operating point of a given distributed communication system at which the spectrum is efficiently used. We do not pretend to solve this tough issue for general distributed multiuser channels. Rather, we  show that it is possible to fully determine such an operating point in one possible scenario which has already been considered in the literature \cite{Cong2008}. This scenario is as follows. We consider an initial communication system comprising two point-to-point communications which use orthogonal channels (say in the frequency domain); a half-duplex relay using a dedicated band is added to the system in order to help the two transmitters to improve their energy-efficiency; the relay implements a pricing mechanism which is directly related to the amount of band used for relaying. The energy-efficiency metric under consideration is a quantity in bit correctly decoded per Joule and is defined as in \cite{Goodman2000}. The situation where each transmitter aims at selfishly maximizing its individual energy-efficiency (with pricing) by allocating bandwidth on the extra channel on which the relay operates has been considered in \cite{Cong2008}~; the solution concept considered therein is the Nash equilibrium (NE), which is shown to be unique but not Pareto efficient. Our goal is to consider another solution concept for this spectrum allocation game of interest namely, the Nash bargaining solution (NBS), motivated by the need to design efficient solutions in distributed wireless networks. Remarkably, such a solution exists for the considered scenario, is unique, and can be implemented in a decentralized manner according to the conjugate gradient algorithm. This confirms the relevance of this approach which has also been adopted in other contexts such as \cite{Larsson2008} (NBS for power allocation games where transmission rates are optimized with no relay and pricing), \cite{Ma2011} (wireless sensors are energy-efficiently coordinated by the Raiffa-Kalai-Smorodinsky solution  to communicate with a unique fusion center), or \cite{Zhang2008} (multiple access channels without pricing are considered). Compared to these references, the present work makes a step towards implementing an efficient solution in a decentralized manner, which is known to be a challenging task \cite{Lasaulce-Tutorial-09}\cite{Lasaulce2011}. The algorithm proposed in this paper is decentralized in the sense of the decision but not in terms of channel state information (CSI), leaving this issue as a non-trivial extension of this work.\\
This paper is structured as follows. In the next section, we introduce the system model as well as the spectrum allocation game. In section \ref{sec:Region}, we analyze the NBS and present the decentralized algorithm. In section \ref{sec:NR}, numerical results are presented and discussed. Concluding remarks are proposed in section \ref{sec:conc}.

\section{Problem statement}
\label{sec:SM}

\subsection{System model}\label{sec:system-model}

The communication system under study is represented in Fig. \ref{fig:system-model}. Source/transmitter $i \in \{1,2\}$ sends a signal $\sqrt{p_i} x_i$ with power $p_i$ over two quasi-static (block fading) links: the link from source $i$ to destination $i$ whose channel gain is $h_{ii} \in \mathbb{C}$ and the one from the source $i$ to the relay $r$ whose channel gain is $h_{ir} \in \mathbb{C}$. The total band associated with those two links is $\omega$ and the extra band allocated by source $i$ to communicate with the relay is denoted by $\omega_i \in [0, \omega]$. The extra band available is precisely that offered by the relay. The relay operates in a half-duplex mode and is assumed to implement an amplify-and-forward (AF) protocol. Time is divided into blocks on which all channel gains are assumed to be fixed. Each block is divided into two sub-blocks \cite{Laneman2004}. Over the first sub-block (first phase), only the source can transmit and the signals received by the destination and relay nodes are given by:
\begin{equation}
\left\{
\begin{array}{lcl}
y_{ii} & = &  h_{ii} \sqrt{p_{i}}  x_{i} + n_{ii},  \\
y_{ir} & = &  h_{ir} \sqrt{p_{i}} x_{i} + n_{ir},
\end{array}
\right.
\label{1}
\end{equation}
where $n_{ii}$ and $n_{ir}$ are (complex) additive white Gaussian noises (AWGN) with mean $0$ and variance $\sigma^{2}$. Following the relevant choice of \cite{Laneman2004}, only the relay is assumed to transmit over the second sub-block (second phase):
\begin{equation}
y_{ri} = h_{ri}  \sqrt{p_{r}} x_{ri}+n_{ri},
\end{equation}
where $h_{ri}\in \mathbb{C}$ is the channel gain between the relay and destination $i$, $n_{ri}\sim \mathcal{N}(0,\sigma^{2})$, and $x_{ri}$ is the transmitted signal from the relay to destination $i$ and, under the assumption made in terms of relaying protocol, expresses as: 
\begin{equation}
x_{ri} = \frac{y_{ir}}{\mid y_{ir} \mid}.
\label{2}
\end{equation}
In the first transmission phase, the signal-to-noise ratio (SNR) associated with the source-destination channel $i$ merely writes as:
\begin{equation}
\gamma_{i,i}  =  \displaystyle{\frac{p_{i}|h_{ii}|^2}{\sigma^{2}}}.
\end{equation}
The SNR associated with the second transmission phase is given by \cite{Cong2008,Laneman2004}:
\begin{equation}
\gamma_{r,i}  =  \displaystyle{\frac{p_{i} p_r |h_{ir}|^2|h_{ri}|^2}{\sigma^{2}(p_{i}
|h_{ir}|^2+p_{r}|h_{ri}|^2+\sigma^{2})}}.
\end{equation}
Interestingly, as proven by \cite{Laneman2004}, when using maximal-ratio combining, the equivalent SNR corresponding to the AF protocol can be written in a simple form if the outage probability is the metric of interest. This writes as: $\gamma_{i,i}^{AF} = \gamma_{i,i} + \gamma_{r,i}$.

\begin{figure}[!h]
\centering
\subfloat[Direct transmission.]{\label{fig:direct}\includegraphics[scale=0.6]{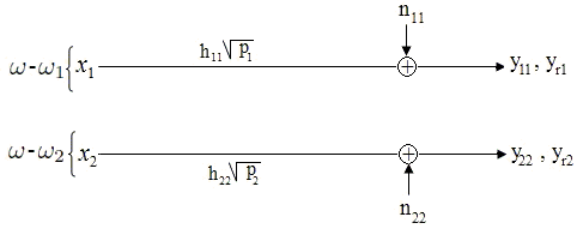}}\\
\subfloat[Cooperative transmission.]{\label{fig:relay}\includegraphics[scale=0.6]{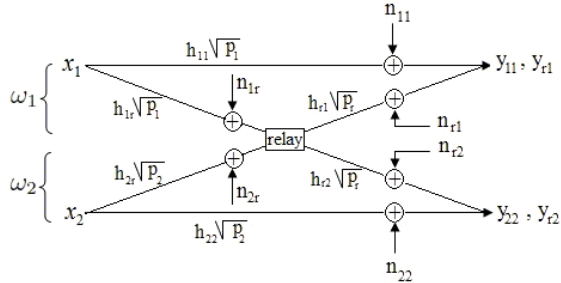}}\\
\caption{System model.}
\label{fig:system-model}
\end{figure}

\subsection{Strategic form of the spectrum allocation game}
\label{sec:game-def}

As motivated in \cite{Cong2008}, the spectrum allocation problem can be modeled by a strategic form game (see e.g., \cite{Lasaulce2011}).
\begin{definition}
The game is defined by the ordered triplet $\mathcal{G} = \bigl(\mathcal{K}, (\mathcal{S}_i)_{i \in \mathcal{K}}, (u_i)_{i \in \mathcal{K}} \bigr)$ where
\begin{itemize}
\item \textit{$\mathcal{K}$ is the set of players. Here, the players of the game are the two sources/transmitters, $\mathcal{K} = \{1,2\}$};
\item \textit{$\mathcal{S}_i$ is the set of actions/strategies. Here, the strategy of source/transmitter $i$ consists in choosing $\omega_{i}$ in its strategy set $\mathcal{S}_{i} = [0, \omega]$};
\item \textit{$u_i$ is the utility function of each user. It is given by:}
\end{itemize}
\begin{eqnarray}
u_{i}(\omega_1, \omega_2) & = & \displaystyle{\alpha \frac{(\omega-\omega_{i})}{p_{i}}f(\gamma_{i,i}) + \alpha \frac{\omega_{i}}{(p_{i}+p_{r})}f(\gamma_{i,i}^{AF})}\nonumber \\
		  & & \displaystyle{-b\left(\sum_{j=1}^{2}\omega_{j}\right)\omega_{i}},
\label{16}
\end{eqnarray}
where $\alpha$ defines the spectral efficiency (in bit/s per Hz), $f : \left[0, +\infty \right) \rightarrow [0,1]$ is a sigmoidal efficiency function which can correspond to the packet success rate or probability of having no outage (emphasizing the link between energy-efficiency and the outage analysis conducted in \cite{Laneman2004}). The parameter $b$ is a (linear) pricing factor.


\end{definition}
 As explained in \cite{Cong2008}, the presence of the factor $b \geq 0$ amounts to imposing a cost to the sources for using the relay; this cost is assumed to be proportional to the relaying band used. The first term of the utility function corresponds to the ratio between the goodput (net rate in bit/s) to the cost in terms of power (in J/s) for the direct link alone (whose band width equals $\omega-\omega_i$), whereas the second term corresponds to the aggregated effects of the direct transmission and the relayed transmission (whose bandwidth equals $\omega_i$). This game is concave in the sense of Rosen and has a unique pure NE (see e.g., \cite{Lasaulce2011}). The main problem is that the NE can be very inefficient as there exist some operating points at which both transmitters have better utilities. This motivates the study of more efficient solutions such as the NBS. The NBS analysis, the design of a simple distributed optimization algorithm to implement it, and proving its relevance in terms of performance constitute the main results of this paper.

\section{Nash bargaining solution analysis}
\label{sec:Region}

The objective of this section is to characterize the NBS of the game $\mathcal{G}$ and to propose a simple distributed optimization algorithm for implementation since the function of interest to optimize can be checked to be strictly concave under certain operation conditions explicated in Sec.~\ref{sec:SC}. 

\subsection{Achievable utility region: Pareto boundary and convexity}
First, we study the properties of the achievable or feasible utility region, which is denoted by $\mathcal{R}$. It is defined as the region formed by all the points whose coordinates are $(u_{1},u_{2})$ that is:
\begin{equation}
\mathcal{R} = \{\left(u_{1},u_{2}\right)|\left(\omega_{1},\omega_{2}\right)\in \left[0,\omega\right]^2\}.
\label{19}
\end{equation}
For a given channel configuration or block of data (i.e., the $h_{ij}$ are given), the region $\mathcal{R}$ is compact \cite{Larsson2008}, which follows from the compactness
of $\mathcal{S}_i$ and the continuity of $u_{i}$. However, it is not always convex. This prevents one from using bargaining theory which is based on the convexity of the achievable utility region. It turns out that, in the problem under consideration, it is relevant to exploit time-sharing (as done in \cite{Larsson2008} for Shannon-rate efficient allocation games on the interference channel), which convexifies the utility region. Indeed, the main idea is to assume that coordination in time is part of the sought solution. The new utility region is:
\begin{equation}
\begin{aligned}
\bar{\mathcal{R}} = &\{\left(\mu u_{1}+(1-\mu)u_{1}',\mu u_{2}+(1-\mu)u_{2}'\right)\\
&|0 \leq \mu \leq 1,\; \left(u_{1},u_{2}\right)\in \mathcal{R},\; \left(u_{1}',u_{2}'\right)\in \mathcal{R}\}.
\end{aligned}
\label{20}
\end{equation}
During a fraction $\mu$ of the time, the users use $\left(\omega_{1},\omega_{2}\right)$ to have $\left(u_{1},u_{2}\right)$. During a fraction $(1-\mu)$ of the time, they use another combination of bandwidths $\left(\omega_{1}',\omega_{2}'\right)$ to obtain $\left(u_{1}',u_{2}'\right)$. Note that this region includes several points of interest. First, it includes $\left(u_{1}^{NE},u_{2}^{NE}\right)$ which is the point corresponding to the unique pure NE of $\mathcal{G}$. Second, it includes the two points for which the sources or transmitters do not exploit the relay at all: $\omega_{1}=0$ or $\omega_{2}=0$. Let $\bar{\mathcal{R}}^{\ast}$ be the Pareto boundary of the convex hull $\bar{\mathcal{R}}$. Fig. \ref{fig:achievableregion} illustrates different operating points as well as the achievable utility region $\mathcal{R}$ and the Pareto boundary $\bar{\mathcal{R}}^{\ast}$. The other elements shown in this figure are defined next.
\begin{figure}[!h]
\centering
\includegraphics[scale=0.5]{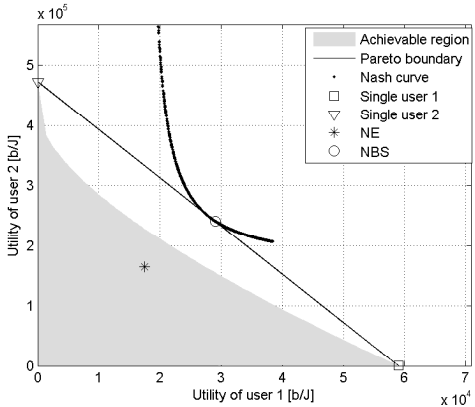}
\caption{The achievable utility region plus some key operating points.}
\label{fig:achievableregion}
\end{figure}
\subsection{Existence and uniqueness analysis of the NBS}
The NBS can be characterized as follows:
\begin{proposition}
 In the spectrum allocation game $\mathcal{G}$, there exists a unique NBS given by:
\begin{equation}
\left(u_{1}^{NBS},u_{2}^{NBS}\right)=\max_{\left(u_{1},u_{2}\right)\in\bar{\mathcal{R}}^{+}} {(u_{1}-u_{1}^{NE})(u_{2}-u_{2}^{NE})},
\label{30}
\end{equation}
where
\begin{equation}
\bar{\mathcal{R}}^{+} = \{\left(u_{1},u_{2}\right)|u_{1}\geq u_{1}^{NE},\; u_{2}\geq u_{2}^{NE}\}.
\label{29}
\end{equation}
\end{proposition}

\begin{proof}
The point $\left(u_{1}^{NE},u_{2}^{NE}\right)$ defines a threat point and can always be reached, which ensures the existence of a solution to the above maximization problem.

Regarding to the uniqueness of the NBS, Nash proved that it holds under certain axioms due to the existence of the convex hull of the achievable region and the threat point, as mentioned in \cite{Larsson2008}. As we have shown that the utility region can be convexified, this solution is also unique.

Finally, the cooperative outcome (NBS) must be invariant to equivalent utility representations, symmetric, independent of irrelevant alternatives and Pareto efficient \cite{Thomson1989}. The NBS is therefore the unique solution resulting from the intersection of the Pareto boundary $\bar{\mathcal{R}}^{\ast}$ with the Nash curve which is defined as (Fig. \ref{fig:achievableregion}):
\begin{equation}
(u_1,u_2) = \arg \max_{(u_1,u_2)} \pi(u_1,u_2)
\label{31}
\end{equation}
where $\pi(u_1,u_2) = (u_{1}-u_{1}^{NE})(u_{2}-u_{2}^{NE})$ is the Nash product function.\end{proof} Since the NBS determination is on the subregion $\bar{\mathcal{R}}^{+}$, we stress that the utilities arising from the NBS are higher than those deduced with NE (see Fig. \ref{fig:achievableregion}).
\subsection{Decentralized algorithm for the NBS determination}
\label{sec:AD}
The proposed algorithm is based on the idea of determining analytically the unique maximum, which is the NBS from the resolution of the following system of equations:
\begin{equation}
\mbox{(I)}\left\{
\begin{array}{ccc}
\displaystyle{\frac{\partial \pi}{\partial \omega_{1}}}& = & 0,\\
\displaystyle{\frac{\partial \pi}{\partial \omega_{2}}}& = & 0.
\end{array}
\right.
\label{32}
\end{equation}
Mathematical resolution of such a system leads to solve two second degree polynomials in $\omega_{i}$ (for $i\in\{1,2\}$), the discriminants of which are fourth degree polynomials in $\omega_{j}$ (for $j\in\{1,2\} \setminus \{i\}$). The study of signs of the discriminants show that expressing the NBS analytically is a difficult task even by exploiting Ferrari and Cardan methods for high degree polynomials resolution. This study shows the interest in: (1) finding decentralized algorithms to compute the NBS; (2) designing distributed procedures to converge towards the NBS. The scope of this paper is about (1) and (2) but with the restriction that distribution is only performed in terms of decision and not in terms of channel state information. Instead of determining the maximum of $\pi$, we propose to find the minimum of $-\pi$, denoted after as $\pi_{m}$, by focusing on the conjugate gradient algorithm. One of the steps of this algorithm consists in determining a parameter denoted as $\beta_{k+1}$ (which is defined next). Accordingly, many methods have been introduced such as: Fletcher-Reeves, Polak-Ribi\`{e}re and Hestenes-Stiefel. Due to the efficiency of its convergence, we focus here on the second method based on calculating the Polak-Ribi\`{e}re parameter \cite{Polak1971}. The spectrum sharing policy is updated in an alternating manner, just like the iterative sequential iterative water-filling algorithm \cite{Yu2002}. However, in contrast with the latter, only the decision is distributed here and global channel state information is needed (through the Hessian matrix).
\vspace{0.3cm}\\
\begin{tabular*}{0.49\textwidth}{l}
  \hline
  \makebox[\linewidth][c]{\textbf{Algorithm $1$: Decentralized determination of the NBS}}\\
  \hline
\textbf{(1)} Set the position of the relay\\
\textbf{(2)} $\omega^{0}=(\omega_{1}^{0},\omega_{2}^{0})$ (frequency initialization)\\
\textbf{(3)} $v_{0} =-\nabla \pi_{m}(\omega^{0})$ (initialization gradient)\\
\textbf{(4)} \textbf{k=0}; \textbf{while $\|v_{k}\|>\epsilon$}\\
\textbf{a.} $t_{k} = \displaystyle{-\frac{g_{k}^{t}v_{k}}{v_{k}^{t}A_{m}v_{k}}}$ (optimal parameter with Newton method\\
where $g_{k}=\nabla \pi_{m}(\omega^{k})$, $A_{m}$ is the Hessian matrix of $\pi_{m}$ and\\
$\omega^{k}=(\omega_{1}^{k},\omega_{2}^{k})$ is the frequency bands at the $k^{th}$ iteration)\\
\textbf{b.} $ \omega^{k+1}= \omega^{k} + t_{k}v_{k}$ (new frequency bands)\\
\textbf{c.} $\omega_{1}^{k+1}=\omega^{k+1}(1)$ and $\omega_{2}^{k+1}=\omega^{k}(2)$ (alternated updates)\\
\textbf{d.} $g_{k+1}=\nabla \pi_{m}(\omega^{k+1})$ (new gradient) \\
\textbf{e.} $\beta_{k+1}=\displaystyle{\frac{g_{k+1}^{t}(g_{k+1}-g_{k})}{g_{k}^{t}g_{k}}} $ (Polak-Ribi\`{e}re parameter)\\
\textbf{f.} $v_{k+1}=-g_{k+1}+\beta_{k+1}v_{k}$ (new descent direction)\\
\textbf{k = k+1}\\
\textbf{end}\\
\textbf{(5)} $(\omega_{1}^{NBS},\omega_{2}^{NBS})=(\omega^{FI}(1),\omega^{FI}(2))$ where \textit{FI} denotes \\
\hspace{0.48cm}\textit{Final Iteration}\\
  \hline
\end{tabular*}

\subsection{Convergence of the algorithm (Strict-concavity analysis of the $\pi$ function)}
\label{sec:SC}
The proposed algorithm is ensured to converge to a NBS if $\pi$ is strictly concave. But this property is not always true. According to the previous study, the function $\pi$ is defined on the subregion $\bar{\mathcal{R}}^{+}$ which is formed by all the utilities $(u_{1},u_{2})$ verifying $u_{i}\geq u_{i}^{NE}$ for all $i \in \{1,2\}$. Such a set can be determined when the NE point is fixed. Though, for each channels values, a NE can be identified. Consequently, the subregion $\bar{\mathcal{R}}^{+}$ depends on the channels values. In the following, for given locations of sources and destinations, we show that there exists a region in which the $\pi$ function is strictly concave.

Proving the strict-concavity of $\pi$ amounts to proving the strict-negativity of the eigenvalues of its corresponding Hessian matrix, which is given by:
\begin{equation}
A=\left(
   \begin{array}{cc}
      a_{11} & a_{12} \\
      a_{21} & a_{22}
   \end{array}
   \right),
\end{equation}
where:
\begin{equation*}
\begin{array}{lcl}
a_{11} & = & \displaystyle{\frac{\partial^2 \pi}{\partial \omega_{1} ^2}}     \\
& = & -2b(u_2-u_2^{NE})-2b\omega_{2}(-\varphi_{1}+\psi_{1}-b(2\omega_{1}+\omega_{2})),\\
a_{22} & = & \displaystyle{\frac{\partial^2 \pi}{\partial \omega_{2} ^2}}     \\
& = & -2b(u_1-u_1^{NE})-2b\omega_{1}(-\varphi_{2}+\psi_{2}-b(2\omega_{2}+\omega_{1})),\\
a_{12} & = & \displaystyle{\frac{\partial^2 \pi}{\partial \omega_{1} \partial \omega_{2}}}     \\
& = & -b(u_2-u_2^{NE})-b(u_1-u_1^{NE})+b^2\omega_{1}\omega_{2}+\\
& & (\varphi_{1}-\psi_{1}+b(2\omega_{1}+\omega_{2}))(\varphi_{2}-\psi_{2}+b(2\omega_{2}+\omega_{1})),\\
a_{21} & = & \displaystyle{\frac{\partial^2 \pi}{\partial \omega_{1} \partial \omega_{2}}}     \\
& = & a_{12},
\end{array}
\label{351}
\end{equation*}
with $\varphi_{i}=\alpha f(\gamma_{i,i})/p_{i}$ and $\psi_{i}=\alpha f(\gamma_{i,i}^{AF})/(p_{i}+p_{r})$ for $i\in\{1,2\}$. Therefore, the eigenvalues are the zeros of the following polynomial:
\newcommand{\tr}{\mbox{tr}}
\begin{equation}
P : \lambda^2 - \lambda \mbox{ } \tr(A) + \det(A).
\end{equation}
If we denote $\Delta$ the discriminant of the polynomial $P$, we can verify merely that $\Delta$ is always positive. Indeed, we have:
\begin{equation}
      \Delta  = (a_{11}-a_{22})^2+4a_{12}a_{21}.
   \label{eqdelta}
\end{equation}
Since $a_{12}=a_{21}$, equation (\ref{eqdelta}) is equivalent to:
\begin{equation}
\Delta =  (a_{11}-a_{22})^2+4(a_{12})^2 \geq 0.
\end{equation}
Therefore, the eigenvalues of $A$, denoted as $\lambda_1$ and $\lambda_2$ are real and are as follows:
\begin{equation}
\left \{
   \begin{array}{lcl}
      \lambda_{1} & = & \displaystyle{\frac{\tr(A)-\sqrt{\Delta}}{2}}, \\
      \lambda_{2} & = & \displaystyle{\frac{\tr(A)+\sqrt{\Delta}}{2}}.
   \end{array}
   \right.
\end{equation}
From these expressions, we study the the strict negativity of the eigenvalues depending on the relay position in a space region $[0,700]\times[0,700]$ m\up{2}. The corresponding simulations (for the same settings considered in section \ref{sec:NR}) are given in Fig. \ref{fig:prob} in which we represent in black the region where the eigenvalues are strictly negative. We assume a standard choice for $f$ for all the numerical results provided in this paper which is $f(x)=(1-e^{-x/2})^{M}$ \cite{Goodman2000} where $M$ is the number of symbols per packet. Therefore, we can deduce that the strict concavity of function $\pi$ is ensured in a region $(x_{r},y_{r})\in[400,550]\times[400,550]$ m\up{2}.
\begin{figure}[!h]
\centering
\includegraphics[scale=0.50]{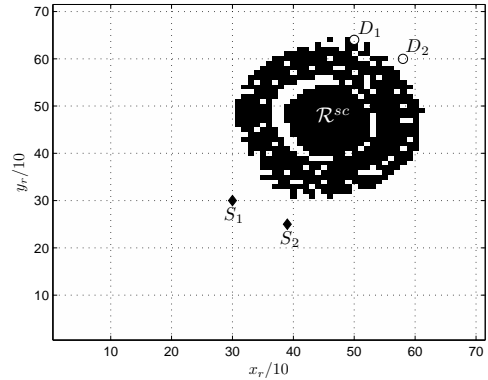}
\caption{Strict-concavity of the $\pi$ function on the disk $\mathcal{R}^{sc}$ when both eigenvalues are strictly negative.}
\label{fig:prob}
\end{figure}

                
\section{Numerical results}
\label{sec:NR}
Here, we implement the NBS and compare it to the NE \cite{Cong2008}. We consider a scenario where the coordinates (in meter) of each source/destination nodes $S_i/D_i$ are as follows: $S_{1}(300,300)$, $D_{1}(500,645)$, $S_{2}(390,257)$ and $D_{2}(590,603)$. The channel gains $|h_{ij}|^{2}$ are given by $0.097/d^{4}$ where $d$ is the distance between the transmitter and the receiver. The noise power and transmission powers of the users and relay are $10^{-13}$ Watt, $0.1$ Watt and $0.08$ Watt respectively, and $\alpha$ is set to $0.8$ bit/s per Hz. The constants $b$ and $M$ are set to $10^{-5}$ and $80$ respectively, while the bandwidth $\omega$ is fixed to $1$ MHz.
\begin{figure}[!h]
\centering
\includegraphics[scale=0.50]{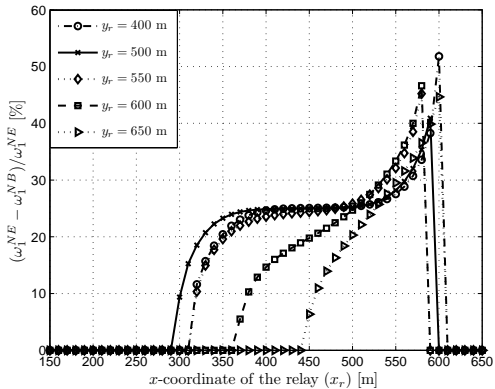}
\caption{Gains in terms of individual bandwidth for user $1$ when operating at the NBS instead of the NE.}
\label{fig:BU1}
\end{figure}
\begin{figure}[!h]
\centering
\includegraphics[scale=0.50]{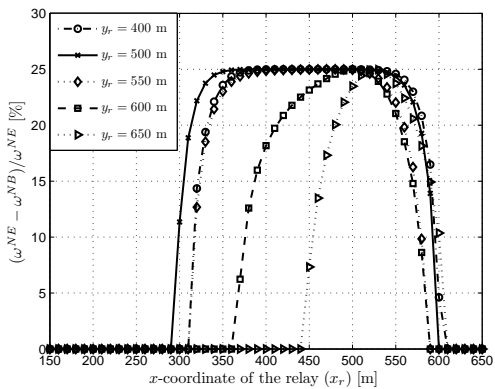}
\caption{Gains in terms of system/total bandwidth by operating at the NBS instead of the NE.}
\label{fig:DBT}
\end{figure}
\begin{figure}[!h]
\centering
\includegraphics[scale=0.50]{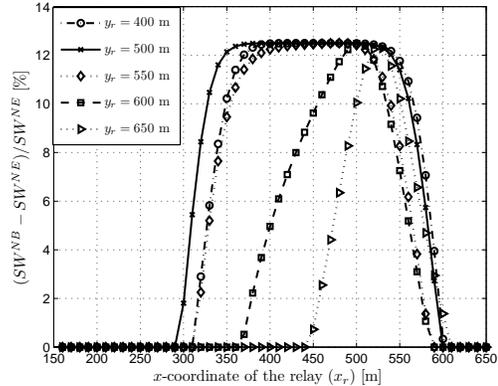}
\caption{Gains in terms of sum energy-efficiency with pricing (social welfare) by operating at the NBS instead of the NE.}
\label{fig:SW}
\end{figure}

Our results highlight that the NBS requires less bandwidth than the NE. Additionally, the energy-efficiency (with pricing) at the NBS is higher than the one at the NE. In Fig. \ref{fig:BU1}, we represent the relative bandwidth gain (NBS vs NE) in $\%$ of user 1 w.r.t. the coordinates of the relay (the relative gain of user 2 shows a similar behavior). Simulations show that maximum gains are obtained when $y_{r}\in [400,550]$ m. Moreover, for different $y_{r}$, there are some regions of $x_{r}$ where the gains vanish. In these regions, the optimum bandwidth with NBS is equal to that with NE. Since user 1 cannot profit from the presence of the relay (when this latter is far from the source), we have $\omega_{1}^{NE}=\omega_{1}^{NBS}=0$. However, user 2 maximizes its utility at the NE (when the relay is much closer to its location) and we have $\omega_{2}^{NE}=\omega_{2}^{NBS}\neq0$. This shows that our analysis provides some insights to an operator who would like to optimize the location of a relay.

In Fig.~\ref{fig:DBT}, we plot the gain in terms of total bandwidth demand. Thus, we deduce that a maximum gain with NBS is reached when the relay is positioned at $(x_{r},y_{r})\in[400,550]\times[400,550]$ m\up{2}. In this region, the total bandwidth demand is reduced to $20-25\%$. The study of the social welfare in Fig. \ref{fig:SW} confirms that the maximum energy-efficiency (with pricing) gain, which is around $10-12\%$, is reached in the same region. Consequently, the results obtained according to the strict-concavity analysis are well confirmed when implementing the conjugate gradient algorithm.

\section{Conclusion}
\label{sec:conc}
This paper studies an efficient solution for a relevant game introduced in \cite{Cong2008} by referring to the NBS. Remarkably, up to a time-sharing argument, the corresponding spectrum allocation game can be checked to possess all the properties to have a unique NBS. Through implementing a conjugate gradient algorithm in a decentralized way, considerable gains of $20-25\%$ can be obtained in terms of used bandwidth. The results reached for the two-user case are very encouraging to extend the case study to larger multi-user systems. Interestingly, our analysis gives some insights into how to deploy some relays for improving a distributed network both from a spectral and energy standpoint. This paper is a first step towards designing fully distributed algorithms (in terms of channel state information) or learning techniques which converge to efficient solutions such as the NBS or Raiffa-Kalai-Smorodinsky solution~; this task is known to be challenging and this paper shows the existence of relevant wireless scenarios where this objective might be reachable.

\bibliographystyle{IEEEbib}
\bibliography{references-3}

\end{document}